\documentclass[preprint,12pt]{elsarticle}

\usepackage{amssymb}

\usepackage{amssymb,amsmath,amsfonts,graphicx}
\usepackage{xspace}
\usepackage{hyperref}
\usepackage{amsthm}
\usepackage{multirow}
\usepackage{hhline}
\usepackage{tikz}
\usetikzlibrary{positioning}

\newcommand{\MTF}{\ensuremath{\operatorname{\textsc{Mtf}}}\xspace} 
\newcommand{\MTFT}{\ensuremath{\operatorname{\textsc{Mtf2}}}\xspace} 
\newcommand{\MTFE}{\ensuremath{\operatorname{\textsc{MtfE}}}\xspace} 
\newcommand{\MTFO}{\ensuremath{\operatorname{\textsc{MtfO}}}\xspace} 
\newcommand{\TS}{\ensuremath{\operatorname{\textsc{Timestamp}}}\xspace} 
\newcommand{\Trans}{\ensuremath{\operatorname{\textsc{Transpose}}}\xspace}
\newcommand{\SPLIT}{\ensuremath{\operatorname{\textsc{Split}}}\xspace}
\newcommand{\BIT}{\ensuremath{\operatorname{\textsc{Bit}}}\xspace}

\newcommand{\AlgMin}{\ensuremath{\operatorname{\textsc{AlgMin}}}\xspace}
\newcommand{\AlgMax}{\ensuremath{\operatorname{\textsc{AlgMax}}}\xspace}

\newcommand{\OPT}{\ensuremath{\operatorname{\textsc{Opt}}}\xspace}
\newcommand{\opt}{\ensuremath{\operatorname{\textsc{Opt}}}\xspace}
\newcommand{\len}{l}
\newcommand{\optsize}{\ensuremath{\operatorname{\textsc{Opt}}(\sigma)}\xspace}

\newcommand{\A}{{\ensuremath{A}}\xspace}
\newcommand{\AP}{{\ensuremath{A'}}\xspace}

\newcommand{\alg}{{\ensuremath{\mathbb{A}}}\xspace}
\newcommand{\algp}{{\ensuremath{\mathbb{A}'}}\xspace}

\newcommand{\algo}[1]{{\small #1}}{}

\newtheorem{theorem}{Theorem}[section]
\newtheorem{lemma}[theorem]{Lemma}

\newtheorem{corollary}[theorem]{Corollary}
\newtheorem{definition}[theorem]{Definition}

\newcommand*\arcdo[2]{%
  \raisebox{-3.2pt}{%
  \begin{tikzpicture}
	\hspace{-3pt}
		\draw (-17pt,11pt) -- (-17pt,7pt);
		\draw(-23pt,7pt) to [out=90,in=90] (-17pt,7pt);
		\draw [->] (-23pt,7pt) -- (-23pt,4.5pt);
		\node at (-0.7,0pt) {$#1#2$};
  \end{tikzpicture}}
	\hspace{-11pt}
	}
	
\newcommand*\arcyek[1]{%
  \raisebox{-3pt}{%
	\begin{tikzpicture}
	\hspace{-3.2pt}
  	\draw [->] (-15pt,11pt) -- (-15pt,4.5pt);
		\node at (-15pt,0pt) {$#1$};
	\end{tikzpicture}}
	\hspace{-12pt}
	}
	
\newcommand*\forwardnull{[a_1 \ldots a_l]}
\newcommand*\backwardnull{[a_l \ldots a_1]}

\newcommand*\forwardzero{[\overset{0}{a_1} \ldots \overset{0}{a_l}]}

\newcommand*\forwardone{[\overset{1}{a_1} \ldots \overset{1}{a_l}]}

\newcommand*\backwardzero{[\overset{0}{a_l} \ldots \overset{0}{a_1}]}

\newcommand*\backwardone{[\overset{1}{a_l} \ldots \overset{1}{a_1}] }

\newcommand*\arrowforone[1]{\xrightarrow[#1]{a_1 \ldots a_l}}
\newcommand*\arrowformult[2]{\xrightarrow[#1]{a_1^{#2} \ldots a_l^{#2}}}
\newcommand*\arrowbackone[1]{\xrightarrow[#1]{a_l \ldots a_1}}
\newcommand*\arrowbackmult[2]{\xrightarrow[#1]{a_l^{#2} \ldots a_1^{#2}}}

\definecolor{Blue}{cmyk}{1,1,0,0}

\journal{Information and Computation}

\begin{document}

\begin{frontmatter}



\title{On the List Update Problem with Advice}


\author[inst1]{Joan Boyar}
\author[inst2]{Shahin Kamali}
\author[inst1]{Kim S. Larsen}
\author[inst3]{Alejandro L\'opez-Ortiz}

\address[inst1]{University of Southern Denmark, Department of Mathematics and Computer Science,
Campusvej 55, 5230 Odense M, Denmark, \\
\emph{$\{$joan,kslarsen$\}$@imada.sdu.dk}}

\address[inst2]{Massachusetts Institute of Technology, Computer Science and Artificial Intelligence Laboratory, 32 Vassar Street, Cambridge, MA 02139, U.S.A., \\
\emph{skamali@mit.edu}}

\address[inst3]{University of Waterloo, School of Computer Science, 200 University Avenue West,
Waterloo, ON N2L 3G1, Canada, \\
\emph{alopez-o@cs.uwaterloo.ca}}


\begin{abstract}
We study the online list update problem under the advice model of computation. Under this model, an online algorithm receives partial information about the unknown parts of the input in the form of some bits of advice generated by a benevolent offline oracle. We show that advice of linear size is required and sufficient for a deterministic algorithm to achieve an optimal solution or even a competitive ratio better than $15/14$. On the other hand, we show that surprisingly two bits of advice are sufficient to break the lower bound of $2$ on the competitive ratio of deterministic online algorithms and achieve a deterministic algorithm with a competitive ratio of $1.\bar{6}$. In this upper-bound argument, the bits of advice determine the algorithm with smaller cost among three classical online algorithms, \TS and two members of the \MTFT family of algorithms. We also show that \MTFT algorithms are $2.5$-competitive.
\end{abstract}


\end{frontmatter}


\section{Introduction}
List update is a well-studied problem in the context of online algorithms. The input is a sequence of requests to items of a list; the requests appear in a sequential and online manner, i.e., while serving a request an algorithm cannot look at the incoming requests. A request involves accessing an item in the list.\footnote{Similar to other works, we consider the \textit{static} list update problem in which there is no insertion or deletion.} To access an item, an algorithm should linearly probe the list; each probe has a cost of 1, and accessing an item in the $i$th position results in a cost of $i$. 
The goal is to maintain the list in a way to minimize the total cost. 
An algorithm can make a \textit{free exchange} to move an accessed item somewhere closer to the front of the list. Further, it can make any number of \textit{paid exchanges}, each having a cost of 1, to swap the positions of any two consecutive items in the list. 

Similar to other online problems, the standard method for comparing online list update algorithms is competitive analysis. The competitive ratio of an online algorithm \A is the maximum ratio between the cost of \A for serving any sequence and the cost of \opt for serving the same sequence. Here, \opt is an optimal offline algorithm. It is known that, for a list of length $l$, no deterministic online algorithm can achieve a competitive ratio better than $2l/(l+1)$ (reported in~\cite{Irani91}); this converges to $2$ for large lists. There are 2-competitive algorithms (hence best possible online algorithms) for the problem; these include \algo{Move-To-Front} (\MTF)~\cite{SleTar85} and \TS~\cite{Albers94A}.

Although competitive analysis has been accepted as the standard tool for comparing online algorithms, there are objections to it. One relevant objection is that assuming a total lack of information about the future is unrealistic in many applications. This is particularly the case for the list update problem when it is used as a method for compression~\cite{BeSlTW86}. In this application, each character of a text is treated as an item in the list, and the text as the input sequence which is parsed (revealed) in a sequential manner. A compression algorithm can be devised from a list update algorithm \A by writing the access cost of \A for serving each character in unary.\footnote{Encodings other than unary correspond to other cost models for list update, and, naturally, encoding positions in binary would improve the compression~\cite{BeSlTW86}. The choice of algorithm is also important and tests indicate that \TS may be a better algorithm for this than \MTF~\cite{AlbMit98}} Hence, the size of the compressed file is roughly equal to the access cost of the list update algorithm. In this application, it is possible to include some partial information about the structure of the sequence (text) in the compressed file, for example, which of three algorithms was used to do the compression. This partial information could potentially be stored using very little space compared to the subsequent savings in the size of the compressed file compared with the original file, due to the availability of the partial information {\cite{KamLopDCC14}}. 

Advice complexity provides an alternative for the analysis of online problems. Under the advice model, the online algorithm is provided with some bits of advice, generated by a benevolent offline oracle with infinite computational power. This reduces the power of the adversary relative to the online algorithm. Variant models are proposed and studied for the advice complexity model~\cite{DobrevKP08,EmekFraKorRos2011,BocKomKra09,BocKomKra11}. Here, we use a natural model from~\cite{BocKomKra09,BocKomKra11} that assumes advice bits are written once on a tape before the algorithm starts, and the online algorithm can access the tape sequentially from the beginning at any time. The advice complexity of an algorithm is then the worst case number of bits read from the tape,
as a function of the length of the input.
Since its introduction, many online problems have been studied under the advice model. These include classical online problems such as paging~\cite{BocKomKra09,HromKraKra10,KommKra11}, $k$-server~\cite{EmekFraKorRos2011,BocKomKra11,RenaRosen11,GuptKam13}, bin packing~\cite{BoyKamSTACS,AngelDurr15}, and various coloring problems~\cite{BiaBoc12,ForKelSte12,SieSprUng13}. 


\subsection{Contribution}
When studying an online problem under the advice model, the first question to answer is how many bits of advice are required to achieve an optimal solution. We show that advice of size $\opt(\sigma)$ is sufficient to optimally serve a sequence $\sigma$, where $\opt(\sigma)$ is the cost of an optimal offline algorithm for serving $\sigma$, and it is linear in the length of the sequence, assuming that the length of the list is a constant. We further show that advice of linear size is required to achieve a deterministic algorithm with a competitive ratio better than $15/14$. 

Another important question is how many bits of advice are required to break the lower bound on the competitive ratio of any deterministic algorithm. We answer this question by introducing a deterministic algorithm that receives two bits of advice and achieves a competitive ratio of at most $1.\bar{6}$. The advice bit for a sequence $\sigma$ simply indicates the best option between three online algorithms for serving $\sigma$. These three algorithms are \TS, \algo{MTF-Odd} (\MTFO) and \algo{MTF-Even} (\MTFE).
\TS inserts an accessed item $x$ in front of the first item $y$ (from the front of the list) that precedes $x$ in the list and was accessed at most once since the last access to $x$. If there is no such item $y$ or $x$ is accessed for the first time, no items are moved. \MTFO (resp.\ \MTFE) moves a requested item $x$ to the front on every odd (resp.\ even) request to $x$. 

Our results indicate that if we dismiss \TS and take the better algorithm between \MTFO and \MTFE, the competitive ratio of the resulting algorithm is no better than 1.75. We also study the competitiveness of \MTFE and \MTFO, and more generally any algorithm that belongs to the family of Move-To-Front-Every-Other-Access (also known as \MTFT algorithms). We show that these algorithms have competitive ratios of 2.5.

\section{Optimal solution}\label{optSol}
In this section, we provide upper and lower bounds on the number of advice bits required to optimally serve a sequence. We start with an upper bound:

\begin{theorem}\label{luAdviceUpper}
Under the advice model, $\opt(\sigma)-n$ bits of advice are sufficient to achieve an optimal solution for any sequence $\sigma$ of length $n$, where $\opt(\sigma)$ is the cost of an optimal algorithm for serving $\sigma$.
\end{theorem}

\begin{proof}
It is known that there is an optimal algorithm that moves items using only a family of paid exchanges called \textit{subset transfer}~\cite{ReinWes96}. In a subset transfer, before serving a request to an item $x$, a subset $S$ of items preceding $x$ in the list is moved (using paid exchanges) to just after $x$ in the list, so that the relative order of items in $S$ among themselves remains unchanged. Consider an optimal algorithm \opt which only moves items via subset transfer. Before a request to $x$ at index $i$, an online algorithm can read $i-1$ bits from the advice tape, indicating (bit vector style) the subset which should be moved behind $x$. Provided with this, the algorithm can always maintain the same list as \opt. The total number of bits read by the algorithm will be at most $\opt(\sigma)-n$. 
\end{proof}

The above theorem implies that for lists of constant size, advice of linear size is sufficient to optimally serve a sequence. We show that advice of linear size is also required to achieve any competitive ratio smaller than $15/14$.

In order to prove this lower bound, we first define a large set of possible
sequences, defined from bit strings, where an online algorithm with a good
competitive ratio must make the right decision many times. This ends up
essentially being a `guess', for every fifth request, as to what the 
next request will be.
Thus, a lookahead of~$1$ (being able to see the next request before making
a decision about the current one, the `weak lookahead' as defined 
in~\cite{Albers98j}),  would be sufficient to perform optimally on these
particular sequences.
Less than linear advice leaves some sequences where the
algorithm does not `guess' well enough.

Consider instances of the list update problem on a list of two items $x$ and $y$ which are defined as follows. Assume the list is ordered as $[x,y]$ before the first request. Also, to make our explanation easier, assume that the length of the sequence, $n$, is divisible by 5. Consider an arbitrary bitstring $B$, of size $n/5$, which we refer to as the \textit{defining bitstring}. Let $\sigma$ denote the list update sequence defined from $B$ in the following manner: For each bit in $B$, there are five requests in $\sigma$, which we refer to as a \textit{round}. We say that a round in $\sigma$ is of type 0 (resp.\ 1) if the bit associated with it in $B$ is 0 (resp.\ 1). For a round of type 0, $\sigma$ will contain the requests $yyyxx$, and for a round of type 1, the requests $yxxxx$.
For example, if $B = 0 1 1 \ldots$, we will have $\sigma = \left\langle yyyxx, yxxxx, yxxxx, \ldots \right\rangle$.

Since the last two requests in a round are to the same item $x$, it makes sense for an online algorithm to move $x$ to the front after the first access. This is formalized in the following lemma. 

\begin{lemma}
For any online list update algorithm \alg serving a sequence $\sigma$ created from a defining bitstring, there is another algorithm whose cost is not more than \alg's cost for serving $\sigma$ and that ends each round with the list in the order $[x,y]$.
\end{lemma}

\begin{proof}
Let $R_t$ denote the first round such that the ordering of the list maintained by \alg is $[y,x]$ at the end of the round. So, \alg incurs a cost of 4 for the last two requests of the round (which are both to $x$) and a cost of 1 for the first request of the next round (which is to $y$). This sums to a cost of 5 for these three requests. Consider an alternative algorithm \algp which moves $x$ to the front after the first access to $x$ in $R_t$. The cost of \algp for the last two requests of $R_t$ is 3. Also, \algp incurs a cost of 2 to access the first request of the next round. Hence, \algp incurs a cost of at most 5, equal to the cost of \alg for these three requests. After the access to $y$ in the second position of the list, \algp can reestablish the same ordering \alg uses from that point (using at most one free exchange). Consequently, the cost of \algp is not more than \alg. Repeating this argument for all rounds completes the proof. 
\end{proof}

Provided with the above lemma, we can restrict our attention to algorithms that maintain the ordering $[x,y]$ at the end of each round. In what follows, by an `online algorithm' we mean an online algorithm with this property. 

\begin{lemma} \label{optCost}
The cost of an optimal algorithm for serving a sequence of length $n$, where the sequence is created from a defining bitstring, is at most $7n/5$.
\end{lemma}

\begin{proof}
Since there are $n/5$ rounds, it is sufficient to show that there is an algorithm which incurs a cost of at most 7 for each round. Consider an algorithm that works as follows: For a round of type 0, the algorithm moves $y$ to the front after the first access to $y$. It also moves $x$ to the front after the first access to $x$. 
Hence, it incurs a cost 2+1+1+2+1 = 7. For a round of type 1, the algorithm does not move any item and incurs a cost of 2+1+1+1+1 = 6. In both cases, the list ordering is $[x,y]$ at the end of the round and the same argument can be repeated for the next rounds. 
\end{proof}

For a round of type 0 (with requests to $yyyxx$), if an online algorithm \alg moves each of $x$ and $y$ to the front after the first accesses, it has cost $7$. If it does not move $y$ immediately, it has cost at least $8$. For a round of type 1 (i.e., a round of requests to $yxxxx$), if an algorithm does no rearrangement, its cost will be $6$; otherwise its cost is at least $7$. To summarize, an online algorithm should `guess' the type of each round and act accordingly after accessing the first request of the round. If the algorithm makes a wrong guess, it incurs a `penalty' of at least 1 unit. This relates our problem to the binary string guessing problem, defined in~\cite{EmekFraKorRos2011,BocHroKom13}.

\begin{definition}[\cite{BocHroKom13}]
The {\em Binary String Guessing Problem with known history ($2$-SGKH)} is the following online problem. The input is a bitstring of length $m$, 
and the bits are revealed one by one. For each bit $b_t$, the online algorithm \alg must guess if it is a $0$ or a $1$. After the algorithm has made a guess, the value of $b_t$ is revealed to the algorithm.
\end{definition}

\begin{lemma}[\cite{BocHroKom13}] \label{servi}
On an input of length $m$, any deterministic algorithm for $2$-SGKH that is guaranteed to guess correctly on more than $\alpha m$ bits, for $1/2 \leq \alpha < 1$, needs to read at least $(1 + (1 - \alpha) \log(1 - \alpha) + \alpha \log \alpha) m$ bits of advice.\footnote{In this paper we use $\log n$ to denote $\log_2(n)$.}
\end{lemma}

We reduce the $2$-SGKH problem to the list update problem: 

\begin{theorem} \label{luAdviceLower}
On an input of size $n$, any algorithm for the list update problem which achieves a competitive ratio of $\gamma$ ($ 1 < \gamma \leq 15/14$) needs to read at least $(1 + (7\gamma-7)\log(7\gamma-7) +(8-7\gamma) \log (8-7\gamma))/5 \cdot n$ bits of advice.
\end{theorem}

\begin{proof}
Consider the $2$-SGKH problem for an arbitrary bitstring $B$. Given an online algorithm \alg for the list update problem, define an algorithm for $2$-SGKH as follows: Consider an instance $\sigma$ of the list update problem on a list of length 2 where $\sigma$ has $B$ as its defining bitstring, 
and run \alg to serve $\sigma$. For the first request $y$ in each round in $\sigma$, \alg should decide whether to move it to the front or not. The algorithm for the $2$-SGKH problem guesses a bit as being 0 (resp.\ 1) if, after accessing the first item requested in the round associated with the bit in $B$, \alg moves it to front (resp.\ keeps it at its position). 
As mentioned earlier, for each incorrect guess \alg incurs a penalty of at least 1 unit, i.e., $\alg \geq \opt + w$, where $w$ is the number of wrong guesses for critical requests. Since \alg has a competitive ratio of $\gamma$, we have $\alg \leq \gamma \opt$. 
Consequently, we have $w \leq (\gamma-1) \optsize$ and by Lemma~\ref{optCost}, $w \leq 7(\gamma-1)/5 \cdot n$. This implies that if \alg has a competitive ratio of $\gamma$, the 2-SGKH algorithm makes at most $7(\gamma-1)/5 \cdot n$ mistakes for an input bitstring $B$ of size $n/5$, i.e., at least $n/5 - 7(\gamma-1)/5 \cdot n = (8-7\gamma) \cdot n/5$ correct guesses. Define $\alpha = 8-7\gamma $, and note that $\alpha$ is in the range $[1/2,1)$ when $\gamma$ is in the range stated in the lemma. By Lemma~\ref{servi}, at least $(1 + (1 - \alpha) \log(1 - \alpha) + \alpha \log \alpha) n/5$ bits of advice are required by such a 2-SGKH algorithm. Replacing $\alpha$ with $8-7\gamma$ completes the proof. 
\end{proof}

Thus, to obtain a competitive ratio better than $15/14$, a linear number of bits of advice is required. For example, to achieve a competitive ratio of $1.01$, at least $0.12 n$
bits of advice are required. Theorems~\ref{luAdviceUpper} and~\ref{luAdviceLower} imply the following corollary.

\begin{corollary}
For any list of fixed length~$n$, $\Theta(n)$ bits of advice are required and sufficient to achieve an optimal solution for the list update problem. 
Also, $\Theta(n)$ bits of advice are required and sufficient to achieve a $1$-competitive algorithm.
\end{corollary}


\section{An algorithm with two bits of advice}\label{bestThree}
In this section we show that two bits of advice are sufficient to break the lower bound of $2$ on the competitive ratio of deterministic algorithms and achieve a deterministic online algorithm with a competitive ratio of $1.\bar{6}$. 
The two bits of advice for a sequence $\sigma$ indicate which of the three algorithms \TS, \algo{MTF-Odd} (\MTFO) and \algo{MTF-Even} (\MTFE), have the lower cost for serving $\sigma$. Recall that \MTFO (resp.\ \MTFE) moves a requested item $x$ to the front on every odd (resp.\ even) request to $x$. We prove the following theorem:

\begin{theorem}\label{thTwoBits}
For any sequence $\sigma$, we have either $\TS(\sigma) \leq 1.\bar{6} \OPT(\sigma)$, $\MTFO(\sigma) \leq 1.\bar{6} \OPT(\sigma)$, or $\MTFE(\sigma) \leq 1.\bar{6}\OPT(\sigma)$.
\end{theorem}

To prove the theorem, we show that for any sequence $\sigma$, $\TS(\sigma) + \MTFO(\sigma) + \MTFE(\sigma) \leq 5 \OPT(\sigma)$. We note that all three algorithms have the \textit{projective property}, meaning that the relative order of any two items only depends on the requests to those items and their initial order in the list (and not on the requests to other items). \MTFO (resp.\ \MTFE) is projective since in its list an item $y$ precedes $x$ if and only if the last odd (resp.\ even) access to $y$ is more recent than the last odd (resp.\ even) access to $x$. In the lists maintained by \TS, item $y$ precedes item $x$ if and only if in the projected sequence on $x$ and $y$, $y$ was requested twice after the second to last request to $x$ or the most recent request was to $y$ and $x$ has been requested at most once. 
Hence, \TS also has the projective property. 

Similar to most other work for the analysis of projective algorithms,\footnote{Almost all existing algorithms for the list update problem are projective; the only exceptions are \Trans, Move-Fraction~\cite{SleTar85}, and \SPLIT~\cite{Irani91}; see~\cite{KamLop13} for a survey.} we consider the \textit{partial cost model}, in which accessing an item in position $i$ is defined to have cost $i-1$. We say an algorithm is \textit{cost-independent} if its decisions are independent of the cost it has paid for previous requests. The cost of any cost-independent algorithm for serving a sequence of length $n$ decreases $n$ units under the partial cost model when compared to the \textit{full} cost model. Hence, any upper bound for the competitive ratio of a cost-independent algorithm under the partial cost model can be extended to the full cost model. 

To prove an upper bound on the competitive ratio of a projective algorithm under the partial cost model, it is sufficient to prove that the claim holds for lists of size 2. 
The reduction to lists of size two is done by applying a \textit{factoring lemma}, which holds for algorithms not using paid exchanges, ensuring that the total cost of a projective algorithm \alg for serving a sequence $\sigma$ can be formulated as the sum of the costs of \alg for serving projected sequences of two items. A projected sequence of $\sigma$ on two items $x$ and $y$ is a copy of $\sigma$ in which all items except $x$ and $y$ are removed. We refer the reader to \cite[p.~16]{BorElY98} for details on the factoring lemma. Since \MTFO, \MTFE, and \TS{} do not use paid exchanges and since they are projective and cost-independent, to prove Theorem~\ref{thTwoBits}, it suffices to prove the following lemma:

\begin{lemma}\label{lemTwoItemsTwoBits}
Under the partial cost model, for any sequence $\sigma_{xy}$ of two items, we have $\MTFO(\sigma_{xy}) + \MTFE(\sigma_{xy}) + \TS(\sigma_{xy}) \leq 5 \cdot \OPT(\sigma_{xy})$.
\end{lemma}

Before proving the above lemma, we study the aggregated cost of \MTFO and \MTFE on certain subsequences of two items. One way to think of these algorithms is to imagine they maintain a bit for each item. On each request, the bit of the item is flipped; if it becomes `0', the item is moved to the front. Note that the bits of \MTFO and \MTFE are complements of each other. Thus, we can think of them as one algorithm started on complementary bit sequences.
We say a list is in state $[ab]_{(i,j)}$ if item $a$ precedes $b$ in the list and the bits maintained for $a$ and $b$ are $i$ and $j$ ($i,j \in \{ 0,1\}$), respectively. To study the value of $\opt(\sigma_{xy})$, we consider an offline algorithm which uses a free exchange to move an accessed item from the second position to the front of the list if and only if the next request is to the same item. It is known that this algorithm is optimal for lists of two items~\cite{ReinWes96}.

\begin{lemma}\label{lemlem1}
Consider a subsequence of two items $a$ and $b$ of the form $\left\langle (ba)^{2i}\right\rangle$, i.e., $i$ repetitions of $\left\langle baba\right\rangle$. Assume the initial ordering is $[ab]$. The cost of each of \MTFO and \MTFE for serving the subsequence is $3i$ (under the partial cost model). Moreover, at the end of serving the subsequence, the ordering of items in the list maintained by at least one of the algorithms is $[ab]$.
\end{lemma}

\begin{proof}
We refer to a repetition of $baba$ as a \textit{round}. We show that \MTFO and \MTFE have a cost of $3$ for serving each round. 
Assume the bits associated with both items are `0' before serving $ baba$. The first request has a cost of 1 and $b$ remains in the second position, the second request has cost 0, and the remaining requests each have a cost of 1. In total, the cost of the algorithm is 3. The other cases (when items have different bits) are handled similarly. Table~\ref{tabl1} includes a summary of all cases. As illustrated in the table, if the bits maintained for $a$ and $b$ before serving $baba$ are $(0,0)$, (0,1), or (1,1), the list order will be $[ab]$ after serving the round. Since both $a$ and $b$ are requested twice, the bits will be also the same after serving $baba$. Hence, in these three cases, the same argument can be repeated to conclude that the list order will be $[ab]$ at the end of serving $(ba)^{2i}$. Since the bits maintained for the items are complements in \MTFE and \MTFO, at least one of them starts with bits $(0,0)$, $(0,1)$, or $(1,1)$ for $a$ and $b$; consequently, at least one algorithm ends up with state $[ab]$ at the end. 
\end{proof}

\begin{small}
\begin{table*}[!t]
\begin{center}\scalebox{0.75}{

  \begin{tabular}{|c|c|c|c|| }
  \hline

	 Bits for $(a,b)$								& Cost for $\left\langle b a b a \right\rangle$		& Orders before accessing items   																 & Final order \\
	\hline
		$(0,0)$													& $1+0+1+1=3$															 				&[$a$\arcyek{b}] [\arcyek{a}$b$] [\arcdo{a}{b}] [\arcdo{b}{a}]		& $[ab]$ \\
		$(0,1)$													&	$1+1+0+1=3$																			&[\arcdo{a}{b}] [$b$\arcyek{a}] [\arcyek{b}$a$] [\arcdo{b}{a}]		& $[ab]$ \\
		$(1,0)$													& $1+0+1+1=3$																			&[$a$\arcyek{b}] [\arcyek{a}$b$] [\arcdo{a}{b}] [b\arcyek{a}] 		& $[ba]$ \\
		$(1,1)$													& $1+1+1+0=3$																			&[\arcdo{a}{b}] [\arcdo{b}{a}] [$a$\arcyek{b}] [\arcyek{a}$b$] 		& $[ab]$ \\
   \hline 		
	\end{tabular}}
\end{center}
\caption{Assuming the initial ordering of items is $[ab]$, the cost of a both \MTFO and \MTFE for serving subsequence $\left\langle baba \right\rangle$ is 3 (under the partial cost model). The final ordering of the items will be $[ab]$ in three of the cases.}
	\label{tabl1}
\end{table*}
\end{small}


\begin{lemma}\label{lemlem2}
Consider a subsequence of two items $a$ and $b$ which has the form $\left\langle baa\right\rangle$. The total cost that \MTFE and \MTFO incur together for serving this subsequence is at most $4$ (under the partial cost model). 
\end{lemma}

\begin{proof}
If the initial order of $a$ and $b$ is $[ba]$, the first request has no cost, and each algorithm incurs a total cost of at most 2 for the other two requests of the sequence. Hence, the aggregated cost of the two algorithms is 4. Next, assume the initial order is $[ab]$.
Assume the bits maintained by one of the algorithms for $a$ and $b$ are (1,0), respectively. As illustrated in Table~\ref{tabl2}, this algorithm incurs a cost of 1 for serving $baa$; the other algorithm incurs a cost of 3. In total, the algorithms incur a cost of 4. In the other case, when bits maintained for $a$ and $b$ are both `0' in one algorithm (consequently, both are `1' in the other algorithm), the total cost of the algorithms for serving $\left\langle baa\right\rangle$ is 3. 
\end{proof}

\begin{small}
\begin{table*}[!t]
\begin{center}
\scalebox{0.73}{
  \begin{tabular}{|c|c|c|c|c|c|}
  \hline
	 Initial order 			&Bits for $(a,b)$					& Cost for 														&  Orders before    																	& Bits and Costs  				& Total cost  \\
											&													&	$\left\langle b a a \right\rangle$	&  accessing items 																	& (other algorithm) 			&(both algs.) \\
	\hline	
	$[ab]$							&(0,0)										&$1 + 0 + 0 =1$												&[$a$\arcyek{b}] [\arcyek{a}$b$] [\arcyek{a}$b$]		&$(1,1) \rightarrow 2$		&$1+2=3$ \\ 
	$[ab]$							&(0,1)										&$1 + 1 + 1 =3$												&[\arcdo{a}{b}] [$b$\arcyek{a}] [\arcdo{b}{a}]			&$(1,0) \rightarrow 1$    &$3+1=4$ \\ 
	$[ab]$							&(1,0)										&$1 + 0 + 0 =1$												&[$a$\arcyek{b}] [\arcyek{a}$b$] [\arcyek{a}$b$]		&$(0,1) \rightarrow 3$		&$1+3=4$ \\
	$[ab]$							&(1,1)										&$1 + 1 + 0 =2$												&[\arcdo{a}{b}] [\arcdo{b}{a}] [\arcyek{a}$b$]			&$(0,0) \rightarrow 1$		&$2+1=3$ \\
	\hline
	$[ba]$							&(0,0) (0,1)							&$\leq 0+1+1=2$												& -																									& $\leq 2$								&$2+2 = 4$ \\
											&(1,0) (1,1)							&																			&																										&													&						\\
	\hline 		
	\end{tabular}	}
	\caption{The total cost of \MTFO and \MTFE for serving a sequence $\left\langle baa \right\rangle$ is at most 4 (under the partial cost model). Note that the bits maintained by these algorithms for each item are complements of each other. }
	\label{tabl2}

\end{center}
\end{table*}
\end{small}

Using Lemmas~\ref{lemlem1} and~\ref{lemlem2}, we are ready to prove Lemma~\ref{lemTwoItemsTwoBits}:

\begin{proof} [Proof of Lemma~\ref{lemTwoItemsTwoBits}, and consequently Theorem~\ref{thTwoBits}]

Consider a sequence $\sigma_{xy}$ of two items $x$ and $y$. We use the \textit{phase partitioning technique} as introduced in \cite{AlvoWe95}, even using
the same partitioning, though later considering more subcases. 
We partition $\sigma_{xy}$ into \textit{phases} which are defined inductively as follows. Assume we have defined phases up until, but not including, the $t$th request ($t\geq 1$) and the relative order of the two items is $[xy]$ before the $t$th request. Then the next phase is of \textit{type 1} and is of one of the following forms ($j \geq 0$ and $k \geq 1$): 
\begin{equation*}
(a)~ x^jyy \ \ \  (b)~ x^j(yx)^kyy  \ \ \ (c)~ x^j(yx)^kx
\end{equation*}
In case the relative order of the items is $[yx]$ before the $t$th request, the phase has type 2 and its form is exactly the same as above with $x$ and $y$ interchanged.
Note that, after two consecutive requests to an item, \TS, \MTFO and \MTFE all have that item in the front of the list. So, after serving each phase, the relative order of items is the same for all three algorithms. This implies that $\sigma_{xy}$ is partitioned in the same way for all three algorithms. To prove the lemma, we show that its statement holds for every phase.

Table~\ref{tabl3} shows the costs incurred by all three algorithms as well as \OPT for each phase. Note that phases of the form (b) and (c) are divided into two cases, depending on whether $k$ is even or odd. 
We discuss the different phases of type 1 separately. Similar analyses, with $x$ and $y$ interchanged, apply to the phases of type 2. Note that before serving a phase of type 1, the list is ordered as $[xy]$ and the first $j$ requests to $x$ have no cost.

\begin{table*}
\begin{center}
\scalebox{0.63}{
\begin{tabular}{|c|c|c|c|c|c|c|}
  \hline
	 \raisebox{-1.5ex}[0ex][0ex]{Phase} 								&\raisebox{-1.5ex}[0ex][0ex]{\AlgMin}																			& \raisebox{-1.5ex}[0ex][0ex]{\AlgMax} 																		& \raisebox{-1.5ex}[0ex][0ex]{\TS}   						 & Sum (\AlgMin+ 	 			& \raisebox{-1.5ex}[0ex][0ex]{\OPT'} & \raisebox{-1.5ex}[0ex][0ex]{$\frac{\textrm{Sum}}{\OPT'}$}  \\
												&																							&																							&										 & \AlgMax + \TS)				&				&						\\
	\hline	
	$x^jyy$								& $1$																					&$2$																					&$2$								 &$5$														&$1$ 		& $5$  \\ 
	\hline
	$x^j(yx)^{2i}yy$			& $\leq 3i+1$																	&$\leq 3i+2$																	& $2\cdot 2i = 4i$		 &$\leq 10i+3$ 									&$2i+1$ & $<5$ \\ 
	\hline
	$x^j(yx)^{2i-2}yxyy$	& $\leq 3(i-1) + 1$														&$\leq 3(i-1) + 1$														& $2(2i-1)$  	   &$\leq 6(i-1)+2+4$							&$2i$   & $<5$ \\
												& $+ \AlgMin(\left\langle xyy \right\rangle)$	&$+ \AlgMax(\left\langle xyy \right\rangle)$	&$=4i-2$						 &$+(4i-2)=10i-2$								&				&      \\
	\hline
	$x^j(yx)^{2i}x$				& $\leq 3i$																		&$\leq 3i+1$																	&$2 \cdot 2i-1$		 &$\leq (6i+1)+(4i-1)$					&$2i$ 	& $\leq 5$ \\
												&																							&																							&$= 4i-1$						 &$=10i$												&				&			\\
	\hline
	$x^j(yx)^{2i-2}yxx$		& $\leq 3(i-1)$																&$\leq 3(i-1)$																& $2\cdot(2i-1)-1$	 & $\leq 6(i-1)+4$							&$2i-1$ & $\leq 5$  \\
												& $+ \AlgMin(\left\langle yxx \right\rangle)$	&$+ \AlgMax(\left\langle yxx \right\rangle)$	& $=4i-3$						 & $+(4i-3) =10i-5$							&				&      \\												
   \hline 
	\end{tabular}}
	\caption{For use in the proof of Lemma~\ref{lemTwoItemsTwoBits}, we list the costs of \MTFO, \MTFE, and \TS for a phase of type 1 (the phase has type 1, i.e., the initial ordering of items is $xy$). The ratio between the aggregated cost of algorithms and the cost of \OPT for each phase is at most 5. \AlgMin (resp.\ \AlgMax) is the algorithm among \MTFO and \MTFE, which incurs less (resp.\ more) cost for the phase. Note that the costs are under the partial cost model. }
	\label{tabl3}
\end{center}
\end{table*}

Consider phases of form (a), $x^jyy$. \MTFO and \MTFE incur a total cost of 3 for serving $yy$ (one of them moves $y$ to the front after the first request, while the other keeps it in the second position). \TS incurs a cost of 2 for serving $yy$ (it does not move it to the front after the first request). So, in total, the three algorithms incur an aggregated cost of 5. On the other hand, \OPT incurs a cost of 1 for the phase. So, the ratio between the sum of the costs of the algorithms and the cost of \OPT is 5.

Next, consider phases of the form (b). \TS incurs a cost of $2k$ for serving the phase; it incurs a cost of 1 for all requests in $(yx)^{2i}$ except the very first request to $x$, and a cost of 1 for serving the second to last request to $y$. Assume $k$ is even and we have $k=2i$ for some $i\geq 1$, so the phase looks like $x^j (yx)^{2i} yy$. By Lemma~\ref{lemlem1}, the cost incurred by \MTFO and \MTFE is $3i$ for serving $(yx)^{2i}$. We show that for the remaining two requests to $y$, \MTFO and \MTFE incur an aggregated cost of at most 3. If the list maintained by any of the algorithms is ordered as $[yx]$ before serving $yy$, that algorithm incurs a cost of 0 while the other algorithm incurs a cost of at most 2 for these requests; in total, the cost of both algorithms for serving $yy$ will be at most $2$. If the lists of both algorithms are ordered as $[xy]$, one of the algorithms incurs a cost of 1 and the other incurs a cost of 2 (depending on the bit they keep for $y$). In conclusion, \MTFO and \MTFE incur a total cost of at most $6i+3$. \TS incurs a cost of $2k=4i$, while \OPT incurs a cost of $2i+1$ for the phase. To conclude, the aggregated cost of all algorithms is at most $10i+3$ compared to $2i+1$ for \opt, and the ratio between them is less than 5. 

Next, assume $k$ is odd and we have $k= 2i-1$, i.e., the phase has the form $x^j (yx)^{2i-2} yxyy$. The total cost of \MTFO and \MTFE for $(yx)^{2i-2}$ is $2 (3(i-1))$ (Lemma~\ref{lemlem1}), the total cost for the next request to $y$ is at most 2, and the total cost for subsequent $xyy$ is at most 4 (Lemma~\ref{lemlem2}). In total, \MTFO and \MTFE incur a cost of at most $6i$ for the phase. On the other hand, \TS incurs a cost of $4i-2$ for the phase. The aggregated cost of the three algorithms is at most $10i-2$ for the phase, while \OPT incurs a cost of $2i$. So, the ratio between sum of the costs of the algorithms and \OPT is less than 5.

Next, consider phases of type 1 and form (c). \TS incurs a cost of $2k-1$ in this case. Assume $k$ is even, i.e., the phase has the form $x^j (yx)^{2i} x$. By Lemma~\ref{lemlem1}, \MTFO and \MTFE each incur a total cost of $3i$ for $(yx)^{2i}$. Moreover, after this, the list maintained for at least one of the algorithms is ordered as $[xy]$. Hence, the aggregated cost of algorithms for the next request to $x$ is at most 1. Consequently, the total cost of \MTFE and \MTFO is at most $6i+1$ for the round. Adding the cost $2k-1 = 4i -1$ of \TS, the total cost of all three algorithms is at most $10i$. On the other hand, \OPT incurs a cost of $2i$ for the phase. So, the ratio between the aggregated cost of all three algorithms and the cost of \OPT is at most 5. Finally, assume $k$ is odd, i.e., the phase has form $x^j (yx)^{2i-2} yxx$. By Lemma~\ref{lemlem1}, \MTFO and \MTFE together incur a total cost of $2 (3(i-1))$ for $x^j (yx)^{2i-2}$. By Lemma~\ref{lemlem2}, they incur a total cost of at most $4$ for $yxx$. In total, they incur a cost of at most $6(i-1)+4$ for the phase. \TS incurs a cost of $4i-3$; this sums up to $10i-5$ for all three algorithms. In this case, \OPT incurs a cost of $2i-1$. Hence, the ratio between the sum of the costs of all three algorithms and \OPT is at most 5.
\end{proof}

In fact, the upper bound provided in Theorem 3 for the competitive ratio of the best algorithm among \TS, \MTFO and \MTFE is tight under the partial cost model. To show this, we make use of the following lemma.

\begin{lemma}\label{Lemlow1}
Consider a sequence $\sigma_{\alpha} = \left\langle x(yxxx \ yxxx)^k \right\rangle$, i.e., a single request to $x$, followed by $k$ repetitions of $(yxxx \ yxxx)$. Assume the list is initially ordered as $[xy]$. We have $\MTFO(\sigma) = \MTFE(\sigma) = 4k$ while $\opt(\sigma) = 2k$ (under the partial cost model).
\end{lemma}

\begin{proof}
We refer to each repetition of $(yxxx \ yxxx)$ as a round. Initially, the bits maintained by \MTFO (resp.\ \MTFE) for $x,y$ are $(1,1)$ (resp.\ (0,0)). After the first request to $x$, the bits of \MTFO (resp.\ \MTFE) change to $(0,1)$ (resp.\ (1,0)) for $x,y$. \MTFO incurs a cost of 3 for the first half of each round; it incurs a cost of 1 for all requests except the last request to $x$. \MTFE incurs a cost of 1 for serving the first half of a round; it only incurs a cost of 1 on the first request $y$. After serving the first half, the list for each algorithm will be ordered as $[xy]$ and the bits maintained by \MTFO (resp.\ \MTFE) for $x,y$ will be $(1,0)$ (resp.\ (0,1)). Using a symmetric argument, the costs of \MTFO and \MTFE for the second half of a round are respectively $1$ and $3$. In total, both \MTFO and \MTFE incur a cost of 4 for each round. After serving the round, the list maintained by both algorithms will be ordered as $[xy]$ and the bits associated with the items will be the same as at the start of the first round. Thus, \MTFO and \MTFE each have a total cost of $4k$ on $\sigma_\alpha$. A summary of actions and costs of \MTFO and \MTFE can be stated as follows (the numbers below the arrows indicate the costs of requests on top, and the numbers on top of $x$ and $y$ indicate their bits):

\begin{footnotesize}
\begin{align*}
[ \overset{0}{x} \overset{1}{y}]  \xrightarrow[\text{1}]{y} [ \overset{0}{y} \overset{0}{x}]  \xrightarrow[\text{1}]{x} [ \overset{0}{y} \overset{1}{x}]  \xrightarrow[\text{1}]{x} [ \overset{0}{x} \overset{0}{y}]  \xrightarrow[\text{0}]{x} [ \overset{1}{x} \overset{0}{y}]  \xrightarrow[\text{1}]{y} [ \overset{1}{x} \overset{1}{y}]   \xrightarrow[\text{0}]{x} [ \overset{0}{x} \overset{1}{y}]  \xrightarrow[\text{0}]{x} [ \overset{1}{x} \overset{1}{y}]  \xrightarrow[\text{0}]{x} [ \overset{0}{x} \overset{1}{y}] \\
[ \overset{1}{x} \overset{0}{y}]  \xrightarrow[\text{1}]{y} [ \overset{1}{x} \overset{1}{y}]  \xrightarrow[\text{0}]{x} [ \overset{0}{x} \overset{1}{y}]  \xrightarrow[\text{0}]{x} [ \overset{1}{x} \overset{1}{y}] \xrightarrow[\text{0}]{x} [ \overset{0}{x} \overset{1}{y}]  \xrightarrow[\text{1}]{y} [ \overset{0}{y} \overset{0}{x}]  \xrightarrow[\text{1}]{x} [ \overset{0}{y} \overset{1}{x}]  \xrightarrow[\text{1}]{x} [ \overset{0}{x} \overset{0}{y}]  \xrightarrow[\text{0}]{x} [ \overset{1}{x} \overset{0}{y}] 
\label{eq1}
\end{align*}
\end{footnotesize}

An optimal algorithm \opt never changes the ordering of the list and has a cost of 2 for the whole round, giving a cost of $2k$ for $\sigma_\alpha$.
\end{proof}

\begin{theorem}\label{lowPartial}
There are sequences for which the costs of all of \TS, \MTFE, and \MTFO are $1.\bar{6}$ times that of $\opt$ (under the partial cost model).
\end{theorem}

\begin{proof}
Consider a sequence $\sigma = \sigma_\alpha \sigma_\beta$ where $\sigma_\alpha = x (yxxx \ yxxx)^{k_\alpha}$ and $\sigma_\beta = (yyxx)^{k_\beta }$. Here, $k_\alpha $ is an arbitrary large integer and $k_\beta = 2k_\alpha$. By Lemma~\ref{Lemlow1}, we have $\MTFO(\sigma_\alpha) = \MTFE(\sigma_\alpha) = 4k_\alpha $ while $\opt(\sigma_\alpha) = 2k_\alpha $. We have $\TS(\sigma_\alpha) = 2k_\alpha$, because it does not move $y$ from the second position. 

Next, we study the cost of \MTFO and \MTFE for serving $\sigma_\beta$. Note that after serving $\sigma_\alpha$, the lists maintained by these algorithms is ordered as $[xy]$ and the bits associated with $x$ and $y$ are respectively $(0,1)$ for \MTFO and $(1,0)$ for \MTFE (see the proof of Lemma~\ref{Lemlow1}). 
We show that for each round $yyxx$ of $\sigma_\beta$, the cost of each of these two algorithms is 3. On the first request to $y$, \MTFO moves it to the front (since the bit maintained for $y$ is 1); so it incurs a cost of 1 for the first requests to $y$. On the first request to $x$, \MTFO keeps $x$ in the second position; hence it incurs a cost of 2 for the requests to $x$. In total, it has a cost of 3 for the round. With a similar argument, \MTFE incurs a cost of 2 for the requests to $y$ and a cost of 1 for the requests to $x$ and a total cost of 3. The list order and bits maintained for the items will be the same at the end
of the round as at the start.
Hence, the same argument can be extended to other rounds to conclude that the cost of both \MTFE and \MTFO for serving $\sigma_\beta$ is $3k_\beta$. On the other hand, \TS incurs a cost of 4 on each round as it moves items to the front on the second consecutive request to them; hence, the cost of \TS for serving $\sigma_\beta$ is $4k_\beta$. An algorithm that moves items in front on the first of two consecutive requests to them will incur a cost of 2 on each round; hence the cost of \opt for serving $\sigma_\beta$ is at most $2k_\beta$. 

To summarize, the cost of each of \MTFO and \MTFE for serving $\sigma$ is $4k_\alpha + 3k_\beta = 10k_\alpha$ while the cost of \TS is $2k_\alpha + 4k_\beta = 10k_\alpha$, and the cost of \opt is $2k_\alpha + 2k_\beta = 6k_\alpha$. As a consequence, all three algorithms have a cost which is $10/6 = 1.\bar{6}$ times that of \opt. 
\end{proof}

Before continuing with lower bounds, we compare the the results
of Theorem~\ref{thTwoBits} with the randomized algorithm COMB~\cite{AlvoWe95},
which chooses to use BIT~\cite{ReWeSl94} with probability $4/5$ and \TS with 
probability $1/5$.
BIT is the randomized algorithm which, for each item in the list, initially
chooses randomly and independently with probability $1/2$ whether that item
should be moved to the front on odd or even accesses to it. BIT is
$1.75$-competitive~\cite{ReWeSl94}. COMB achieves
a competitive ratio of $1.6$, so for any request sequence $I$, either \TS
must achieve a performance ratio of $1.6$ compared to \OPT, or there must
be some setting of the randomized bits for BIT which achieves a ratio
of $1.6$. This immediately gives an online algorithm with advice achieving
the ratio $1.6$ and using $\ell+1$ bits of advice, one bit to specify
\TS or BIT and $\ell$ bits for BIT.

The phase partitioning in the proof of Theorem~\ref{thTwoBits} uses
more subcases than the proof for COMB, since it cannot assume independence
of the times when $x$ and $y$ are moved to the front. It might be
tempting to try to obtain an online algorithm with advice which achieves
the ratio $1.6$, as COMB does.
The idea would be to create a randomized algorithm which uses
\MTFE and \MTFO each with probability $2/5$ and \TS with probability $1/5$,
and then change that to an algorithm using advice instead.
The proof of Theorem~\ref{lowPartial}
shows that this does not work, at least in the partial cost model. 
Consider the sequence $\sigma_\alpha$. Since both \MTFE and \MTFO have cost
$4k_\alpha$, and \TS and \OPT both have cost $2k_\alpha$, the performance
ratio one achieves with this weighting is $1.8$.

On the other hand, note that \MTFE and \MTFO are equivalent to BIT
in the cases where all the random bits are identical, either all zero
or all one. Thus, using Theorem~\ref{thTwoBits}, one obtains a randomized
algorithm with less randomness than COMB (choosing with equal probabilities
between \MTFE, \MTFO, and \TS) with a competitive ratio of $1.\bar{6}$.
In addition, by Theorem~\ref{theorem-mtfe-mtfo-upper-lower} below,
if only one bit of randomness
is used for BIT, to decide whether \MTFE or \MTFO is used, the resulting
algorithm is $2$-competitive. The lower bound proven on the competitiveness of
this algorithm using only one bit of advice is~$1.75$.

It is clear in general that a $c$-competitive randomized algorithm using 
$b(n)$ random bits for sequences of length $n$ automatically gives a
$c$-competitive algorithm with advice using $b(n)$ bits of advice.
Advice can be more powerful, though. It was shown in~\cite{BocKomKra11}
that for minimization problems, for all $\epsilon>0$,
the existence of a $c$-competitive randomized algorithm implies the
existence of a $(1+\epsilon)c$-competitive algorithm using at most
$\lceil\log n\rceil+2\lceil \log\lceil\log n\rceil\rceil + log\left( \lfloor
\frac{\log(m(n)}{\log(1+\epsilon)} \rfloor\right)+1$ bits of advice, where
$m(n)$ is the number of possible inputs of length $n$.

The lower bound from Theorem~\ref{lowPartial} cannot easily be extended to the full cost model. In what follows, we provide, for the full cost model, a lower bound of 1.6 for the competitive ratio of the best algorithm among \TS, \MTFE, and \MTFO. We start with the following lemma:

\begin{lemma}\label{lower3full1}
Consider a list of $\len$ items which is initially ordered as $[a_1, a_2, \ldots, a_\len]$. 
Consider the following sequence of requests
with an $m$-fold repetition: 
\[\sigma_{\beta} = \left\langle (a_1,a_2,...,a_\len,a_1^2,a_2^2,...,a_{\len}^2,a_\len,a_{\len-1},...,a_1,a_\len^2,a_{\len-1}^2,...,a_1^2)^m \right\rangle.\]
Then for large $\len$, we have $\MTFO(\sigma_\beta) = \MTFE(\sigma_\beta) = m\cdot (3.5 \len^2 + o(\len^2))$ while $\TS(\sigma_\beta) = m\cdot (2 \len^2+ o(\len^2))$ (under the full cost model).
\end{lemma}

\begin{proof}
Define a phase to be a subsequence of requests which forms one of the $m$ repetitions in $\sigma_\beta$. We calculate the costs of the algorithms for each phase. Note that each phase contains an even number of requests to each item.
 Also, if $i<j$, so item $a_i$ precedes item $a_j$ in the initial ordering of
the list, then, in each phase, $a_i$ is requested twice after the last request 
to $a_j$. 
Each algorithm moves $a_i$ in front of $a_j$ on the first or second of these 
requests. Thus, the state of the list maintained by all algorithms is the same 
as with the initial ordering after serving a phase.

Each of the three algorithms incurs a cost of $l(l+1)/2$ for serving $a_1, a_2, \ldots, a_l$ at the beginning of a phase.  \MTFO moves items to the front,
reversing the list, but \MTFE and \TS do not move the items. 
For serving the subsequent requests to $a_1^2,a_2^2,...,a_{\len}^2$, \MTFO incurs 
a cost of $2l^2$ since it does not move items to the front on the first 
of the two consecutive requests to an item, but on the second request.
\MTFE and \TS move to the front 
at the first of the consecutive requests and incur a cost of $l(l+1)/2 + l$ 
(the second request is to front of the list). At this point, for all three
algorithms, the list is in the reverse of the initial ordering
 since for $i<j$ there have been two consecutive requests to $a_j$ after the 
last request to $a_i$. Also, the bits maintained by \MTFE and \MTFO are flipped
compared to the beginning of the phase (since there have been three requests to each
item). Thus, for the second half of the list, \MTFE and \MTFO reverse roles.
For the next requests to $a_l, a_{l-1}, \ldots, a_1$, only \MTFE reverses the
list, and each of the three algorithms incurs a cost of $l(l+1)/2$.
Consequently, for the remaining requests to $a_l^2, a_{l-1}^2, \ldots, a_1^2$, 
\MTFE incurs a cost of $2l^2$, while \MTFO and \TS each incur a cost of 
$l(l+1)/2 + 2l$.

To summarize, the costs of both \MTFO and \MTFE for each phase is 
$3.5 l^2 + o(l^2)$, while the cost of \TS is $2l^2+o(l^2)$. 
The actions and costs of the algorithms can be summarized as following (as before, the numbers below arrows indicate the cost for serving the sequence on top, and the numbers on top of items indicate the bits maintained by \MTFO and \MTFE).
The three lines correspond to \MTFE, \MTFO, and \TS, respectively.

\begin{footnotesize}
\begin{align*}
\forwardzero \arrowforone{l^2/2 + o(l^2)} \forwardone \arrowformult{l^2/2 + o(l^2)}{2} \backwardone \arrowbackone{l^2/2 + o(l^2)} \forwardzero  \arrowbackmult{2l^2+o(l^2)}{2} \forwardzero \\
\forwardone \arrowforone{l^2/2 + o(l^2)} \backwardzero \arrowformult{2l^2 + o(l^2)}{2} \backwardzero \arrowbackone{l^2/2 + o(l^2)} \backwardone \arrowbackmult{l^2/2+o(l^2)}{2} \forwardone\\
\forwardnull \arrowforone{l^2/2 + o(l^2)} \forwardnull \arrowformult{l^2/2 + o(l^2)}{2} \backwardnull \arrowbackone{l^2/2 + o(l^2)} \backwardnull \arrowbackmult{l^2/2+o(l^2)}{2} \forwardnull
\end{align*}
\end{footnotesize}
\end{proof}

The sequence $\sigma_\beta$ of the above lemma shows that using one bit of
advice to decide between using \MTFE and \MTFO gives a competitive ratio 
of at least 1.75, but \TS serves $\sigma_\beta$ optimally. Next, we 
introduce sequences for which \TS performs significantly worse than both \MTFO 
and \MTFE.

\begin{lemma}\label{lower3full2}
Consider a list of $\len$ items which is initially ordered as $[a_1, a_2, \ldots, a_\len]$. 
Consider the following sequence of requests: 
\[\sigma_{\gamma} =  \left\langle (a_l^3,a_2^3,...,a_1^3)^{2s}\right\rangle.\]
Assuming that $\len$ is sufficiently large, we have $\MTFO(\sigma_\gamma) = \MTFE(\sigma_\gamma) = s(3 \len^2 + o(\len^2))$, while $\TS(\sigma_\gamma) = 
s(4 \len^2+ o(\len^2))$ and $\opt(\sigma_\gamma) = s(2 \len^2+ o(\len^2))$ (under the full cost model).
\end{lemma}

\begin{proof}
Define a phase to be two consecutive repetitions of the subsequence in parentheses. We calculate the costs of the algorithms for each phase. Note that 
there are an even number of requests in each phase, and for $i<j$, there are
(actually more than) two consecutive requests to $a_i$ after the last request 
to $a_j$. So the list orderings and bits maintained by \MTFO and \MTFE are 
the same for each algorithm before and after serving each phase. Similarly, 
after serving the first half of a phase (the subsequence in parentheses), 
the lists of all three algorithms are the same as the initial ordering.

An optimal algorithm applies the \MTF strategy and incurs a cost of 
$2l^2 + 4l$. More precisely, for serving each half of the phase, it incurs a cost of $l^2$ for the first of three consecutive requests to each item, and a total cost of $2l$ for the second and third requests. \TS moves items to the front on the second of three consecutive requests. In each half of a phase, it incurs a total cost of $2l^2$ for the first two requests to items and a cost of $l$ for the third requests. In total, it incurs a cost of $4l^2+2l$ for each phase. For the first half of the phase, \MTFO moves items to front on the first request to each item, while \MTFE does so on the second requests. Hence, \MTFO and \MTFE respectively incur a cost of $l^2+2l$ and $2l^2+l$ for the first half. For the second half, the bits maintained by the algorithms are flipped, while the list ordering is the same as the initial ordering. Hence, \MTFO and \MTFE respectively incur a cost of $2l^2+l$ and $l^2+2l$ for the second half. In total the costs of each of \MTFO and \MTFE for each phase is $3l^2 + 3l$. Since the cost of all algorithms are the same for all phases, the statement of the lemma follows. 
The actions and costs of the algorithms for each phase can be summarized as follows:

\begin{align*}
\MTFE: & \forwardzero \arrowbackmult{2l^2+o(l^2)}{3} \forwardone \arrowbackmult{l^2+o(l^2)}{3} \forwardzero \\
\MTFO: &\forwardone \arrowbackmult{l^2+o(l^2)}{3} \forwardzero  \arrowbackmult{2l^2+o(l^2)}{3} \forwardone \\
\TS:  &\forwardnull \arrowbackmult{2l^2+o(l^2)}{3} \forwardnull  \arrowbackmult{2l^2+o(l^2)}{3} \forwardnull \\
\OPT: & \forwardnull \arrowbackmult{l^2+o(l^2)}{3} \forwardnull  \arrowbackmult{l^2+o(l^2)}{3}
\end{align*}
\end{proof}

We use the above two lemmas to prove the following theorem.
We remark that for the partial cost model, a stronger result is
proven in~\cite{AmGaSt13}, which establishes that $1.6$ is a lower bound
for any projective algorithm.

\begin{theorem}\label{lowfull}
The competitive ratio of the best algorithm among \MTFE, \MTFO, and \TS is at least $1.6$ under the full cost model.
\end{theorem}

\begin{proof}
Consider the sequence $\sigma = \sigma_\beta \sigma_\gamma$, i.e., the concatenation of the sequences $\sigma_\beta$ and $\sigma_\gamma$ as defined in Lemmas~\ref{lower3full1} and~\ref{lower3full2}. Recall that these sequences consist of $m$ and $s$ phases, respectively. In defining $\sigma$, consider values of $s$ which
are multiples of $3$, and let $m=2s/3$.

Assume the initial ordering is also the same as the one stated in the lemmas, and recall that the state of the algorithms (list ordering and bits of \MTFO and \MTFE) are the same at the end of serving $\sigma_\beta$.
The costs of each of \MTFE and \MTFO for serving $\sigma_\beta$ is $3.5 l^2 m +o(l^2 m) = \frac{7}{3} l^2 s + o(l^2 s)$ (by Lemma~\ref{lower3full1}), while they incur a cost of $3l^2 s + o(l^2s)$ for $\sigma_\gamma$ (by Lemma~\ref{lower3full2}). In total, each of these two algorithms incurs a cost of $\frac{16}{3} l^2 s + o(l^2 s)$ for $\sigma$. On the other hand, \TS incurs a cost of $2l^2 m + o(l^2 m) = \frac{4}{3} l^2 s + o(l^2 s)$ for $\sigma_\beta$ and a cost of $4 l^2 s + o(l^2 s)$ for $\sigma_\gamma$. In total, its cost for $\sigma$ is $\frac{16}{3} l^2 s + o(l^2 s)$. Note that all three algorithms have the same costs for serving $\sigma$. The cost of \opt for serving $\sigma_\beta$ is at most $2 l^2 m + o(l^2 m) = \frac{4}{3} l^2 s + o(l^2 s)$ (by Lemma~\ref{lower3full1})),  while it has a cost of $2 l^2 s + o(l^2 s)$ for serving $\sigma_\gamma$. In total, the cost of \opt is $\frac{10}{3} l^2 s + o(l^2 s)$. Comparing this with the cost of $\frac{16}{3} l^2 s + o(l^2 s)$ of the three algorithms, we conclude that the minimum competitive ratio is at least $1.6$.
\end{proof}

Thus, the competitive ratio of the best of the three algorithms, \MTFO, \MTFE, and \TS, is at least $1.6$ and at most $1.\bar{6}$. We concluded after
Lemma~\ref{lower3full1} that the competitive ratio of the better of
\MTFO and \MTFE is at least 1.75.
Here, we show that the competitive ratio of the better algorithm among \MTFO and \MTFE is at most 2, using the potential function method.

\begin{lemma}\label{lemUp1}
For any sequence $\sigma$ of length $n$, we have
\[\MTFO(\sigma) +  \MTFE(\sigma) \leq 4 \OPT(\sigma).\]
\end{lemma}

\begin{proof}
Consider an algorithm \A $\in \{\MTFO, \MTFE\}$. At any time $t$ (i.e., before serving the $t$th request), we say a pair $(a,b)$ of items forms an \textit{inversion} if $a$ appears before $b$ in the list maintained by \A while $b$ appears before $a$ in the list maintained by \OPT. We define the \textit{weight} of an inversion $(a,b)$ to be $1$, if the bit maintained by \A for $b$ is 1, and 2 otherwise. Intuitively, the weight of an inversion is the number of accesses to the latter of the two items in \A's list before the item is moved to the front and the inversion disappears.

We define the potential, $\Phi_t$, at each time $t$ to be the total weight of the inversions in the list maintained by \MTFO plus the total weight of the inversions in the list maintained by \MTFE. 

We consider the events that involve costs and change the potential function.
An {\em online} event is the processing of a request by both
\MTFO and \MTFE.
An {\em offline} event is \OPT making a paid exchange. The latter is not
directly associated with a request, and we define the cost of
\MTFO and \MTFE in connection with this event to be zero, but there
may be a change in the potential function.

For an event at time $t$, we define the amortized cost $a_t$ to be the total cost paid by \MTFO and \MTFE together for processing the request (if any), plus the increase in potential due to that processing, i.e., $a_t = \MTFO_t + \MTFE_t + \Phi_t - \Phi_{t-1}$. So the total cost of \MTFO and \MTFE for serving a sequence $\sigma$ is $\sum_t{a_t} - (\Phi_{last} - \Phi_{0})$.  The maximum possible value of $ \Phi_{last} - \Phi_{0}$ is independent of the length of the sequence. Hence, to prove the competitiveness of \MTFO and \MTFE together, it is enough to bound the amortized cost relative to $\OPT$'s cost. Let $\OPT_t$ be the cost paid by \OPT at event $t$. To prove the lemma, it suffices to show that for each event, we have $a_t \leq 4 \OPT_t$.

Note that one may assume that \OPT only does paid exchanges, no free ones~\cite{ReinWes96}.
Consider \MTFO and \MTFE for an {\em online event}.
 Let \A be the algorithm that moves $y$ to the
 front, while \AP is the other algorithm, i.e., the one that keeps it at its current position. Assume \A accesses $y$ at index $k$ while \AP finds it at index $k'$. Also, let $j$ denote the index of $y$ in the list maintained by \opt. 

We first show that the contribution by \A to the amortized cost is at most $k - (k-j) + 2j = 3j$. The first term ($k$) is the access cost for \A. 
Before moving $y$ to front, there are at least  $k-j$ inversions with $y$ for \A involving items which occur before $y$, each having a weight of 1 (since the bit of $y$ in \A has been 1 as it moves $y$ to front). All these inversions are removed after moving $y$ to front. This gives the second term in the amortized cost, i.e., $-(k-j)$. Moving $y$ to the front creates at most $j$ new inversions, each having a weight of at most 2, which results in a total increase of $2j$ in the potential.
Next, we show that the contribution by \AP to the amortized cost is at most $k' - (k'-j) = j$. This is because, after accessing $y$ at index $k'$, there are at least $k'-j$ inversions with $y$ for \AP involving items which occur before $y$.
Since \AP does not move $y$ to the front, the bit of $y$ was $0$, i.e., all these inversions had weight 2. After the access, the bit of $y$ becomes 1 and the weights of these inversions decreases 1 unit. To summarize, the amortized cost $a_t$ is at most $3j + j = 4j$. Since \opt accesses $y$ at index $j$, we have $\opt_j = j$ and consequently $a_t \leq 4 \OPT_t$.

Next, consider an {\em offline event} where \opt makes a paid exchange. In doing so, it incurs a cost of 1 and $\opt_t = 1$. This single exchange might create an inversion in the list of \MTFO and an inversion in the list of \MTFE. Each of these inversions have a weight of at most 2. So, the total increase in the potential is at most 4, i.e., $a_t \leq 4$. Consequently, $a_t \leq 4 \OPT_t$.
\end{proof}

The above lemma implies that the better algorithm between \MTFO and \MTFE has a competitive ratio of at most 2. By Lemma~\ref{lower3full1}, such an algorithm has a competitive ratio of at least 1.75. 

\begin{theorem}
\label{theorem-mtfe-mtfo-upper-lower}
The competitive ratio of the better algorithm between \MTFO and \MTFE is at least $1.75$ and at most $2$.
\end{theorem}


\section{Analysis of Move-To-Front-Every-Other-Access}
In the previous sections, we have used \MTFE and \MTFO to devise algorithms with better competitive ratios. These algorithms are Move-To-Front-Every-Other-Access algorithms (also called \MTFT algorithms). In this section, we study the competitive ratio of these algorithms. 
In \cite[Exercise 1.5]{BorElY98}, it is stated that \MTFT is 2-competitive (throughout, by `\MTFT', we mean `an algorithm that belongs to the family of \MTFT algorithms'). The same statement is repeated in~\cite{BacElY02,KamLop13}. It was first observed in~\cite{Grune03} that \MTFT is in fact \emph{not} 2-competitive. There, the author proves a lower bound of 7/3 for the competitive ratio of \MTFT, and claims that an upper bound of 2.5 can be achieved. Here, we show that the competitive ratio of \MTFT is 2.5, and it is tight.

\begin{lemma}\label{lower3full3}
The competitive ratio of Move-To-Front-Every-Other-Access algorithms is at least $2.5$ (under both partial and full cost models). 
\end{lemma}

\begin{proof}
We prove the lemma for \MTFO and later extend it to other Move-To-Front-Every-Other-Access algorithms.

Consider a list of $\len$ items, initially ordered as $[a_1, a_2, \ldots, a_\len]$. 
Consider the following sequence of requests: 
\[\sigma_{\delta} = \left\langle (a_1,a_2,...,a_\len,a_1^3,a_2^3,...,a_{\len}^3,a_\len,a_{\len-1},...,a_1,a_\len^3,a_{\len-1}^3,...,a_1^3)^m \right\rangle.\]
We show that asymptotically, the cost of \MTFO is 2.5 times the cost of \opt. 
Similar to our other lower bound proofs, we define a phase as a subsequence of requests which forms one of the $m$ repetitions in $\sigma_\delta$. Note that each phase contains an even number of requests to each item.
 Also, if $i<j$, meaning that item $a_i$ precedes item $a_j$ in the initial ordering of
the list, then, in each phase, $a_i$ is requested three times after the last request 
to $a_j$. 
\MTFO moves $a_i$ in front of $a_j$ due to these requests. Thus, the state of the list maintained by the algorithm is the same 
as the initial ordering after serving a phase.

Both \MTFO and \OPT incur a cost of $l(l+1)/2$ for serving $a_1, a_2, \ldots, a_l$ at the beginning of a phase. \MTFO moves items to the front and reverses the list, but \OPT does not move the items. 
For serving the subsequent requests to $a_1^3,a_2^3,...,a_{\len}^3$, \MTFO incurs a cost of $2l^2+l$ since it moves items to the front on the second 
of the three consecutive requests to an item. \OPT moves to the front at the first of the consecutive requests and incurs a cost of $l(l+1)/2 + 2l$ 
(the second and third requests are to the front of the list). At this point, for both algorithms, the list is in the reverse of the initial ordering, while the bits maintained by \MTFO are the same as in
the beginning of the phase,
since there have been four requests to each item, i.e., they are all 1.

For the next requests to $a_l, a_{l-1}, \ldots, a_1$, \MTFO reverses the
list and incurs a cost of $l(l+1)/2$. \opt has the same cost and does not move the items.
Consequently, for the remaining requests to $a_l^3 , a_{l-1}^3, \ldots, a_1^3$, 
\MTFO incurs a cost of $2l^2+l$, while \OPT incurs a cost of $l(l+1)/2 + 2l$.

To summarize, in each phase, the cost of \MTFO is
$5 l^2 + o(l^2)$, while the cost of \OPT is $2l^2+o(l^2)$. 
The actions and costs of the algorithms can be summarized as follows. 

\begin{footnotesize}
\begin{align*}
&\MTFO:\\
&\forwardone \arrowforone{l^2/2 + o(l^2)} \backwardzero \arrowformult{2l^2 + o(l^2)}{3} \backwardone \arrowbackone{l^2/2 + o(l^2)} \forwardzero \arrowbackmult{2l^2+o(l^2)}{3} \forwardone\\[2ex]
&\OPT:\\
&\forwardnull \arrowforone{l^2/2 + o(l^2)} \forwardnull \arrowformult{l^2/2 + o(l^2)}{3} \backwardnull \arrowbackone{l^2/2 + o(l^2)} \backwardnull \arrowbackmult{l^2/2+o(l^2)}{3} \forwardnull
\end{align*}
\end{footnotesize}

We can extend the above lower bound to show that \MTFE is at least 2.5-competitive.  In doing so, consider the sequence $\left\langle (a_1, a_2, \ldots, a_\len) \sigma_{\delta} \right\rangle$. Note that after serving the subsequence in parentheses, all bits maintained by \MTFE become 1, and the same analysis as above holds for serving $\sigma_\delta$. More generally, for any initial setting of the bits maintained by a Move-To-Front-Every-Other-Access algorithm, we can start a sequence with a single request to each item having bit 0. After this subsequence, all bits are 1 and we can continue the sequence with $\sigma_{\delta}$ to prove a lower bound of 2.5 for the competitive ratio of these algorithms. The starting subsequence adds an extra term of at most $l^2/2 + o(l^2)$ to the costs of both \opt and the algorithm.  This extra term can be ignored for sufficiently long sequences, i.e., when the value of $m$ is asymptotically larger than $l$. 
\end{proof}

As mentioned earlier, an upper bound of 2.5 for the competitive ratio of \MTFT was claimed earlier~\cite{Grune03}. Here we include the proof for completeness since it does not appear to have ever been published.

\begin{lemma}\label{up22}
The competitive ratio of any algorithm \A which belongs to the Move-To-Front-Every-Other-Access family of algorithms is at most $2.5$.
\end{lemma}

\begin{proof}

We prove the statement for the partial cost model. Since \A has the projective property and is cost-independent, the upper bound argument extends to the full cost model. 
Consider a sequence $\sigma_{xy}$ of two items $x$ and $y$. As before, we use the phase partitioning technique and partition $\sigma_{xy}$ into phases as in the proof of Lemma~\ref{lemTwoItemsTwoBits}. Recall that a phase
ends with two consecutive requests to the same item in $\sigma_{xy}$. A phase has type 1 (respectively~2) if the relative order of $x$ and $y$ is $[xy]$ (respectively $[yx]$) at the beginning of the phase. Recall that a phase of type 1 has one of the following three forms ($j \geq 0$ and $k \geq 1$): 
\begin{equation*}
(a)~ x^jyy \ \ \  (b)~ x^j(yx)^kyy  \ \ \ (c)~ x^j(yx)^kx
\end{equation*}
A phase of type 2 has exactly the same form as above with $x$ and $y$ interchanged. To prove the lemma, we show that its statement holds for every \emph{two consecutive} phases. First, we consider each phase separately and show that the cost of \MTFT is at most 2 times that of \opt for all phases except a specific phase type that we call a \emph{critical phase}.
Table~\ref{tabl333} shows the costs incurred by \MTFT and \OPT for each phase. Note that phases of the form (b) and (c) are divided into two and three cases,
respectively. The last row in the table corresponds to a critical phase. 
We discuss the different phases of type 1 separately. Similar analyses, with $x$ and $y$ interchanged, apply to the phases of type 2.

Note that before serving a phase of type 1, the list is ordered as $[xy]$ and the first $j$ requests to $x$ have no cost.
Consider phases of the form (a), $x^jyy$. \MTFT{} incurs a total cost of at most 2 for serving $yy$ and \OPT incurs a cost of 1. So, the ratio between the costs of \MTFT and \OPT is at most 2.

Next, consider phases of the form (b) with $k=2i$ ($i$ is a positive integer). By Lemma~\ref{lemlem1}, the cost incurred by \MTFT is at most $3i$ for serving $(yx)^{2i}$. For the remaining two requests to $y$, \MTFT incurs a cost of at most 2. In total, the cost of \MTFT is at most $3i+2$ compared to $2i+1$ for \opt, and the ratio between them is less than 2. 

Next, assume $k$ is odd and $k= 2i-1$, i.e., the phase has the form $x^j (yx)^{2i-2} yxyy$. The total cost of \MTFT for $(yx)^{2i-2}$ is at most $3(i-1)$ (Lemma~\ref{lemlem1}), and its cost for the next requests to $yxyy$ is at most 4. In total, it incurs a cost of at most $3i+1$ for the phase, which is no more than twice the cost $2i$ of \opt.

Next, consider phases of the form (c). Assume $k$ is even, i.e., the phase has the form $x^j (yx)^{2i} x$. By Lemma~\ref{lemlem1}, \MTFT incurs a cost of at most $3i$ for $(yx)^{2i}$ and a cost of at most 1 for the single request to $x$. The cost of the algorithm will be $3i+1$ compared to $2i$ of \opt, and the ratio between them is no more than 2. Next, assume $k$ is odd, i.e., the phase has the form $x^j (yx)^{2i} yxx$ or $x^j yxx$ (as before, $i$ is a positive integer). In the first case, by Lemma~\ref{lemlem1}, \MTFT incurs a cost of $3i$ for $x^j (yx)^{2i}$ and a cost of at most 3 for $yxx$. This sums to $3i+3$ while \opt incurs a cost of $2i+1$; the ratio between these two is no more than 2. If the phase has the form $x^j yxx$, we refer to it as a critical phase. The cost of \MTFT for such a phase can be as large as 3 while \opt incurs a cost of 1. However, we show that the cost of \MTFT in two consecutive phases is no more than twice the cost of \opt.

Consider two consecutive phases in $\sigma_{xy}$. If none of the phases are critical, the cost of \MTFT is at most twice that of \opt in both phases and we are done.  Assume one of the phases is critical while the other phase is not.  Let $\opt_1$ and $\opt_2$ denote the cost of \opt for the critical and non-critical phases, respectively. We have $\opt_1 \leq \opt_2$ because \opt incurs a cost of 1 for critical phases and a cost of at least 1 for other phases (see Table~\ref{tabl333}). The cost of \MTFT for serving the two phases is at most $3\opt_1 + 2 \opt_2$ which is no more than $2.5(\opt_1+\opt_2) $ (since $\opt_1 \leq \opt_2$). This implies that the cost of \MTFT is no more than 2.5 more than that of \opt for the two phases. Finally, assume both phases are critical. Thus, they form a subsequence $(x^jyxx) (x^{j'}yxx)$ in $\sigma_{xy}$. \MTFT moves $y$ to the front for exactly one of the two requests to $y$. Thus, it incurs a cost of 1 for one of the phases and a cost of at most 3 for the other phase. In total, its cost is no more than 4, while \opt incurs a cost of 2 for the two phases.
\end{proof}

\begin{table*}
\begin{center}
\scalebox{0.77}{
\begin{tabular}{|c|c|c|c|c|}
  \hline
	 Phase 														& \MTFT 																							 		& \OPT' & ratio  \\ 
	\hline	
	$x^jyy$									&$\leq 2$																					&$1$ 		& $\leq 2$  \\ 
	\hline									
	$x^j(yx)^{2i}yy$				&$\leq 3i+2$																	&$2i+1$ & $<2$ \\ 
	\hline
	$x^j(yx)^{2i-2}yxyy$		&$\leq 3(i-1) + 4$														&$2i$   & $\leq 2$ \\
	\hline
	$x^j(yx)^{2i}x$					&$\leq 3i+1$																	&$2i$ 	& $\leq 2$ \\
	\hline
	$x^j(yx)^{2i}yxx$			&$\leq 3i + 3$																&$2i+1$ & $\leq 2$  \\
	\hline
	$x^jyxx$								&$3$																					&$1$ 	& $3$ \\
   \hline 
	\end{tabular}}
	\caption{The costs of \MTFT and \opt for a phase of type 1 (i.e., the initial ordering of items is $xy$). The ratio between the cost of \MTFT and \OPT for each phase, except the critical phase (the last row), is at most 2. }
	\label{tabl333}
\end{center}
\end{table*}

From Lemmas~\ref{lower3full3} and~\ref{up22}, we conclude the following theorem:

\begin{theorem}\label{upperth}
The competitive ratio of Move-To-Front-Every-Other-Access algorithms is $2.5$.
\end{theorem}


\section{Concluding remarks}
It is generally assumed that the offline oracle that generates advice bits has unbounded computational power. We used this assumption when we showed that $\opt(\sigma)$ bits are sufficient to achieve an optimal solution in Section~\ref{optSol}. However, for the algorithm introduced in Section~\ref{bestThree}, the advice bits can be generated in polynomial time. 
Table \ref{summa} provides a summary of the provided bounds for the competitive ratio of different algorithms.

The offline version of the list update problem is known to be NP-hard~\cite{Ambuh00}. In this sense, our algorithm can be seen as a linear-time approximation algorithm with an approximation ratio of at most $1.\bar{6}$; this is, to the best of our knowledge, the best deterministic offline algorithm for the problem. 
As mentioned earlier, the randomized algorithm COMB~\cite{AlvoWe95}, which
is $1.6$-competitive, implies the existence of an online algorithm achieving a competitive ratio of at most $1.6$ when provided a linear (in the length of the list) number of advice bits. However, from a practical point of view, it is not clear how an offline oracle can smartly generate such bits of advice. Moreover, our results (Theorem~\ref{upperth}) indicate that, regardless of how the initial bits are generated, algorithm \BIT has a competitive ratio of $2.5$ against adaptive adversaries. 
This follows since an adaptive adversary can learn the original random bits
from the behavior of \BIT by requesting all items once, and then give requests to change $0$-bits to $1$-bits.
This initial subsequence has constant length (proportional to the length
of the list). After this, with a renaming of items based on their current order
in the list, the adversary can treat \BIT as \MTFO and give $\sigma_{\delta}$ 
from the proof of Lemma~\ref{lower3full3}.

We proved that the competitive ratio of the best algorithm among \MTFE, \MTFO, and \TS is at least $1.6$ and at most $1.\bar{6}$. Similarly, for the better algorithm between \MTFE and \MTFO the competitive ratio is between 1.75 and 2. It would be interesting to close these gaps.
\begin{table*}
\begin{center}
\scalebox{0.67}{
\begin{tabular}{|c|c|c|c|}
\hline
\multirow{3}{*}{Algorithm} 
 & \multicolumn{2}{|c|}{Lower Bound}&
 {Upper Bound} \\
& partial cost model & full cost model & (both models)\\
  \hline 
\multirow{3}{*}{ Best of \MTFO, \MTFE, and \TS} & & & \\&$1.\bar{6}$&1.6 &$1.\bar{6}$ \\ & (Theorem \ref{lowPartial}) &  (Theorem \ref{lowfull} &(Theorem \ref{thTwoBits})\\ \hline

\multirow{3}{*}{ Better of \MTFO and \MTFE} & & & \\&$2$& 1.75 &$2$ \\ & (Lemma \ref{Lemlow1})  & 	  (Lemma \ref{lower3full1}) &(Lemma \ref{lemUp1})\\ \hline

\multirow{3}{*}{\MTFT (\MTFO, \MTFE, etc.)} &  & & \\  &2.5 &2.5&$2.5$ \\ & (Lemma \ref{lower3full3}) & (Lemma \ref{lower3full3})& (Lemma \ref{up22})\\ \hline
\hline
\end{tabular}}
\label{summa}
\caption{Summary of proved lower and upper bounds for the competitive ratio of different algorithms. }
\end{center}
\end{table*}

\section*{Acknowledgements}
The authors would like to thank the referees for inquisitive and
constructive comments.
The first and third authors were supported in part by the Villum Foundation
and the Danish Council for Independent Research, Natural Sciences.

\section*{References}
  \bibliographystyle{elsarticle-harv} 
  \bibliography{confs,online}





\end{document}